\documentclass[a4paper]{article}

\usepackage[utf8]{inputenc}
\usepackage{authblk}
\usepackage{amsmath}
\usepackage{amsthm}
\usepackage{enumitem}
\usepackage[autosize]{dot2texi}
\usepackage{tikz}
\usetikzlibrary{shapes,arrows}
\usepackage{subcaption}
\usepackage{xcolor}
\usepackage[toc]{appendix}

\title{Pointwise Metrics for Clustering Evaluation}
\author[1]{Stephan van Staden}
\affil[1]{Google Switzerland GmbH}
\date{May 2024}
\newtheorem{theorem}{Theorem}
\newtheorem{lemma}[theorem]{Lemma}

\begin{document}

\maketitle

\begin{abstract}
This paper defines pointwise clustering metrics, a collection of metrics for characterizing the similarity of two clusterings. These metrics have several interesting properties which make them attractive for practical applications. They can take into account the relative importance of the various items that are clustered. The metric definitions are based on standard set-theoretic notions and are simple to understand. They characterize aspects that are important for typical applications, such as cluster homogeneity and completeness. It is possible to assign metrics to individual items, clusters, arbitrary slices of items, and the overall clustering. The metrics can provide deep insights, for example they can facilitate drilling deeper into clustering mistakes to understand where they happened, or help to explore slices of items to understand how they were affected. Since the pointwise metrics are mathematically well-behaved, they can provide a strong foundation for a variety of clustering evaluation techniques. In this paper we discuss in depth how the pointwise metrics can be used to evaluate an actual clustering with respect to a ground truth clustering.
\end{abstract}


\section{Introduction}

Clustering is the partitioning of a set of items into separate groups, called clusters. The items in each cluster should typically be similar, while the items from different clusters should be different. Although that sounds simple, there are significant challenges in many practical applications of clustering. For example, the criteria of what makes items similar or different might be complex and require human judgement. A clustering algorithm must then approximate that somehow. Moreover, if billions of items must be clustered, then there are typically also constraints such as computing time and cost.

In such complex settings, there is often no optimal clustering algorithm, and developers can experiment with many ideas to improve the quality of the clustering while satisfying the resource constraints. Understanding the resource constraints of a run of an algorithm is fairly straightforward. Understanding the resulting clustering itself, in particular its quality, can be challenging. Yet that is very important for effective development and for the consumers of the clustering.

There are several techniques for evaluating clusterings. In this paper we focus mostly on the evaluation of a clustering with respect to a ground truth clustering. A ground truth clustering consists of a small set of items that humans partitioned into ideal clusters. Once the human work is finished, the ground truth clustering can be stored and used repeatedly to evaluate actual clusterings. Such evaluations are fast; they provide developers with quick feedback. The downside is that the assessment does not consider the items that are not mentioned in the ground truth clustering. Other more expensive evaluation techniques can be used to fill that gap, but typically only after the developer is satisfied with the evaluation results of the ground truth clusterings.

We introduce the \textit{pointwise} clustering metrics, a rich set of metrics that can measure clustering similarity and quality. The pointwise metrics have several properties that make them attractive for practical applications:

\begin{itemize}
    \item The ability to specify the relative importance of each item. In many applications, not all items to be clustered are equally important.
    \item The metrics are intuitive and simple to understand. The definitions use standard set-theoretic notions. They characterize aspects that are important for practical applications such as cluster homogeneity and completeness.
    \item The ability to assign metrics to individual items, clusters, arbitrary slices of items, and the overall clustering. The metrics can provide deep insights, for example they can facilitate drilling deeper into clustering mistakes to understand where they happened, or exploring slices of items to understand how they were affected.
    \item The pointwise metrics provide a strong mathematical foundation for a variety of clustering evaluation techniques. The definitions are elegant and mathematically well-behaved, which makes them an ideal basis for obtaining statistical estimates of the quality of huge clusterings, for example. In this paper we discuss in depth how the pointwise metrics can be used to evaluate an actual clustering with respect to a ground truth clustering.
\end{itemize}

\section{Clusterings and clustering algorithms}

Given a finite set of items $S$ and an equivalence relation $r$ (i.e. a binary relation that is reflexive, symmetric and transitive), a \textit{cluster} is an equivalence class of $S$ with respect to $r$, and a \textit{clustering} is the set of all clusters, i.e. a partitioning of $S$ into its equivalence classes.

The equivalence relation $r$ can often be thought of as being induced by a labeling of items that indicate some classification: suppose there is a function $\mathrm{label} : S \rightarrow \mathrm{L}$, which assigns every item a label from a discrete set of labels $\mathrm{L}$, then $(i_1, i_2) \in r  \Longleftrightarrow  \mathrm{label}(i_1) = \mathrm{label}(i_2)$ is its induced equivalence relation.

In practical applications, the set of items $S$ can be very large and the ideal equivalence relation is not fully known. Humans can consider a pair of items and say whether they are equivalent or not, but since that does not scale to billions of items, we have only very sparse information about the ideal relation. The main job of a clustering algorithm in such a setting is to approximate the ideal equivalence relation. This is typically done by 1) deciding which items might be related (also called `blocking'), and 2) deciding which of these are equivalent according to a computable function that imitates the human judgements. The design space of clustering algorithms is consequently huge, and it becomes desirable to be able to evaluate the clustering results to determine which algorithm and configuration to prefer in practice.

\section{Ground truth clusterings and pointwise metrics}

A ground truth clustering partitions a set of items $S$ into ideal clusters. This is normally done by humans, so the set $S$ tends to be relatively small. For an item $i \in S$, we denote the set of the members of the cluster containing $i$ as $\mathit{IdealCluster}_S(i)$. It always holds that $i \in \mathit{IdealCluster}_S(i)$.

The pointwise metrics take into account the relative importance of the various items. The relative importance is specified by providing a \textit{weight} for each item. Exactly how the weight is determined is application-specific; the pointwise metrics simply require that each weight must be a positive real number. We denote the weight of item $i$ by $\mathit{weight}(i)$, and use the shorthand $\mathit{weight}(I) = \sum_{i \in I} \mathit{weight}(i)$ for a set of items $I$.

We would like to be able to use the ground truth clustering to evaluate a given \textit{actual} clustering of items $S^\prime$, where $S^\prime$ might contain billions of items. For an item $i \in S^\prime$, we denote the set of the members of the actual cluster containing $i$ as $\mathit{ActualCluster}_{S^\prime}(i)$.

If $S$ and $S^\prime$ have no item in common, then an evaluation isn't meaningful because the two clusterings talk about completely different items. Otherwise, we evaluate how much the two clusterings agree by focusing on the set of items that they have in common, namely the items $i$ in $T = S \cap S^\prime$, for which we define:

\begin{align}
    \mathit{IdealCluster}(i) &= \mathit{IdealCluster}_S(i) \cap T \\
    \mathit{ActualCluster}(i) &= \mathit{ActualCluster}_{S^\prime}(i) \cap T
\end{align}

Henceforth, we consider only the set of weighted items $T$ and its two clusterings specified by the functions $\mathit{IdealCluster}$ and $\mathit{ActualCluster}$ respectively.

From the perspective of each item $i \in T$, we can partition $T$ into four sets as follows:

\begin{align}
\mathit{TruePositives}(i) &= \mathit{IdealCluster}(i) \cap \mathit{ActualCluster}(i) \\
\mathit{FalsePositives}(i) &= \mathit{ActualCluster}(i) \setminus \mathit{IdealCluster}(i) \\
\mathit{FalseNegatives}(i) &= \mathit{IdealCluster}(i) \setminus \mathit{ActualCluster}(i) \\
\mathit{TrueNegatives}(i) &= T \setminus [\mathit{TruePositives}(i) \cup \mathit{FalsePositives}(i) \cup \mathit{FalseNegatives}(i)]
\end{align}

This characterization leads immediately to the 2x2 confusion matrix from the perspective of $i$:

\begin{align}
\mathit{TP}(i) &= \mathit{weight}(\mathit{TruePositives}(i)) \\
\mathit{FP}(i) &= \mathit{weight}(\mathit{FalsePositives}(i)) \\
\mathit{FN}(i) &= \mathit{weight}(\mathit{FalseNegatives}(i)) \\
\mathit{TN}(i) &= \mathit{weight}(\mathit{TrueNegatives}(i))
\end{align}

The 2x2 confusion matrix can be used to compute various metrics of interest, for example:

\begin{align}
\mathit{Precision}(i) &= \frac{\mathit{TP}(i)}{\mathit{TP}(i) + \mathit{FP}(i)} \\
\mathit{Recall}(i) &= \frac{\mathit{TP}(i)}{\mathit{TP}(i) + \mathit{FN}(i)} \\
\mathit{Accuracy}(i) &= \frac{\mathit{TP}(i) + \mathit{TN}(i)}{\mathit{TP}(i) + \mathit{TN}(i) + \mathit{FP}(i) + \mathit{FN}(i)}
\end{align}

Note that these definitions are the standard ones for 2x2 confusion matrices. Here, they characterize the difference in the two clusterings from the perspective of each item $i$.

It is also instructive to view the 2x2 confusion matrix as a Venn diagram, which makes the relationship with the clustering situation very direct. That is done in Figure~\ref{confusion_matrix_venn_diagram}. The Venn diagram makes it easy to see how the standard definitions of the Jaccard Distance and the Jaccard Index applies to individual items:

\begin{align}
\mathit{JaccardDistance}(i) &= \frac{\mathit{FN}(i) + \mathit{FP}(i)}{\mathit{FN}(i) + \mathit{TP}(i) + \mathit{FP}(i)} \\
\mathit{JaccardIndex}(i) &= \frac{\mathit{TP}(i)}{\mathit{FN}(i) + \mathit{TP}(i) + \mathit{FP}(i)}
\end{align}

As usual, $\mathit{JaccardDistance}(i) = 1 - \mathit{JaccardIndex}(i)$.

\begin{figure*}[t!]
\centering
\begin{tikzpicture}[fill=gray]
\draw (-1,0) circle (1.5)
      (-1,1.5)  node [text=black,above] {Ideal}
      (1,0) circle (1.75)
      (1,1.75)  node [text=black,above] {Actual}
      (-4,-3) rectangle (4,3)
      (-0.1,3) node [text=black,above] {$\mathit{weight}(T)$}
      % (-0.1,-1) node [text=black,above] {$i$}
      (-0.1,0) node [text=black] {$\mathit{TP}(i)$}
      (-1.6,0) node [text=black] {$\mathit{FN}(i)$}
      (1.6,0) node [text=black] {$\mathit{FP}(i)$}
      (-0.1,-2.3) node [text=black] {$\mathit{TN}(i)$}
      ;
\end{tikzpicture}
\caption{The Venn diagram of the 2x2 clustering confusion matrix from the perspective of item $i$. The item $i$ is always in the intersection of $ \mathit{IdealCluster}(i)$ and $\mathit{ActualCluster}(i)$. Hence $\mathit{weight}(i) \subseteq \mathit{TP}(i)$, so the the weight of $i$ is contained in the intersection, which is always non-empty. The left circle is labeled with Ideal, which is a shorthand for $\mathit{weight}(\mathit{IdealCluster}(i))$. The right circle is labeled with Actual, which is a shorthand for $\mathit{weight}(\mathit{ActualCluster}(i))$.}\label{confusion_matrix_venn_diagram}
\end{figure*}
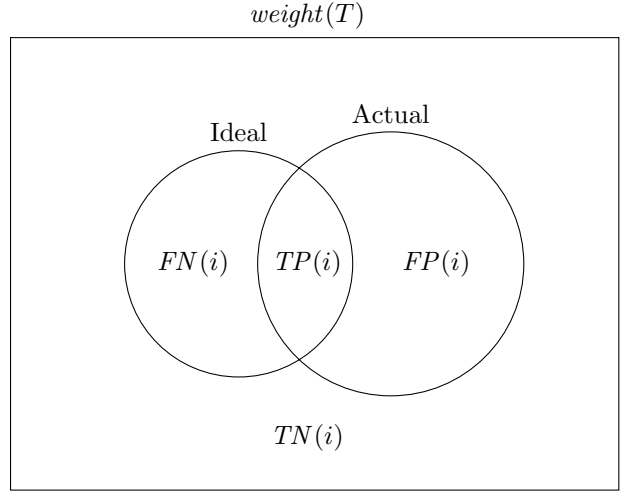

The pointwise formulation makes it easy to assign metrics, such as $\mathit{Precision}$, $\mathit{Recall}$ and $\mathit{JaccardDistance}$, to arbitrary sets of items. Given $I \subseteq T$, we define:

\begin{align}
\mathit{Precision}(I) &= \frac{\sum_{i \in I} \mathit{weight}(i) \mathit{Precision}(i)}{\mathit{weight}(I)} \\
\mathit{Recall}(I) &= \frac{\sum_{i \in I} \mathit{weight}(i) \mathit{Recall}(i)}{\mathit{weight}(I)} \\
\mathit{JaccardDistance}(I) &= \frac{\sum_{i \in I} \mathit{weight}(i) \mathit{JaccardDistance}(i)}{\mathit{weight}(I)}
\end{align}

So $\mathit{Precision}(I)$ is the weighted average $\mathit{Precision}$ of the items in $I$, i.e. the expected $\mathit{Precision}$ of an item in $I$. Taking $I = T$, we obtain an overall $\mathit{Precision}$ metric. We can also obtain $\mathit{Precision}$ metrics for particular slices of items. And we can obtain $\mathit{Precision}$ metrics for individual clusters. The latter can be useful, for example, to see which actual clusters have the worst $\mathit{Precision}$, or which ideal clusters have the worst $\mathit{Recall}$. Working with expected values (i.e. weighted averages) makes the metrics simple to understand and intuitively meaningful.

Other metrics can get the same treatment as $\mathit{Precision}$ did in the previous paragraph. We mention $\mathit{Precision}$ in particular, because it characterizes the homogeneity of clusters. It is often important for practical applications to use clusters that don't mix together too many unrelated items. $\mathit{Recall}$ is the counterpart metric that measures the completeness of clusters. The two go hand-in-hand: we can always easily improve one, but then the other will typically suffer. For example, we can obtain 100\% $\mathit{Precision}$ by putting each item in its own cluster, but then $\mathit{Recall}$ will suffer. Or we can put all the items together in a single cluster and obtain 100\% $\mathit{Recall}$, but then $\mathit{Precision}$ will suffer. Because of their practical usefulness, $\mathit{Precision}$, $\mathit{Recall}$ and $\mathit{JaccardDistance}$ will feature quite heavily in the rest of this paper.

To make this more concrete, Figure~\ref{small_clustering_example} shows how the pointwise metrics work for a toy example. We note that:
\begin{itemize}
    \item The $\mathit{Precision}$ of $i_1$ measures the weight fraction of items in $i_1$'s actual cluster that are shared with $i_1$'s ideal cluster. \\
    The actual cluster of $i_1$ has weight 4, and it shares $\{i_1\}$, with weight 1, with the ideal cluster of $i_1$. Hence $\mathit{Precision}(i_1) = 1/4$.
    \item The $\mathit{Recall}$ of $i_1$ measures the weight fraction of items in $i_1$'s ideal cluster that are shared with $i_1$'s actual cluster. \\
    The ideal cluster of $i_1$ has weight 3, and it shares $\{i_1\}$, with weight 1, with the actual cluster of $i_1$. Hence $\mathit{Recall}(i_1) = 1/3$.
    \item The $\mathit{Precision}$ of the actual cluster $\{i_1, i_3\}$ is the expected $\mathit{Precision}$ of its items. \\
    Item $i_1$ has weight 1 and its $\mathit{Precision}$ is 1/4. \\
    Item $i_3$ has weight 3 and its $\mathit{Precision}$ is 3/4. \\
    Hence the $\mathit{Precision}$ of the actual cluster $\{i_1, i_3\}$ is $\frac{1 \cdot \frac{1}{4} + 3 \cdot \frac{3}{4}}{1 + 3} = \frac{5}{8}$.
    \item Overall $\mathit{Precision}$ is the expected $\mathit{Precision}$ over all items in $T$. For the toy example, we see it is 75\%.
\end{itemize}

One might alternatively, but equivalently, report $\mathit{OverMergeRate}$ and $\mathit{UnderMergeRate}$ metrics. The formulae are simple:

\begin{align}
\mathit{OverMergeRate} &= 1 - \mathit{Precision} \\
\mathit{UnderMergeRate} &= 1 - \mathit{Recall}
\end{align}

With these definitions, it is straightforward to report $\mathit{OverMergeRate}$ and $\mathit{UnderMergeRate}$ metrics on the item, cluster, slice, and overall levels. While that might be more intuitive for some people, we will mostly use $\mathit{Precision}$ and $\mathit{Recall}$ in this paper.

\begin{figure*}[t!]
    \centering
    \begin{subfigure}[t]{0.5\textwidth}
        \centering
        \begin{dot2tex}[dot,mathmode]
            graph {
        		subgraph cluster_1 {
        		    style=rounded
        			node [color=yellow, style=filled]
        			i2 [label=i_2]
        			i1 [label=i_1]
        		}
        		subgraph cluster_2 {
        		    style=rounded
        			node [color=cyan, style=filled]
        			i3 [label=i_3]
        		}
        	}
        \end{dot2tex}
        \caption{The ideal clustering of three items.}
    \end{subfigure}%
    ~
    \begin{subfigure}[t]{0.5\textwidth}
        \centering
        \begin{dot2tex}[dot,mathmode]
            graph {
        		subgraph cluster_1 {
        		    style=rounded
        			node [color=cyan, style=filled]
        			i3 [label=i_3]
        			node [color=yellow, style=filled]
        			i1 [label=i_1]
        		}
        		subgraph cluster_2 {
        		    style=rounded
        			node [color=yellow, style=filled]
        			i2 [label=i_2]
        		}
        	}
        \end{dot2tex}
        \caption{The actual clustering of the items.}
    \end{subfigure}
    ~
    \par\bigskip
    \begin{subfigure}[t]{\textwidth}
    \centering
    \begin{tabular}{|c|c|}
        \hline
              & $\mathit{weight}$ \\
        \hline
        $i_1$ & 1.0 \\
        $i_2$ & 2.0 \\
        $i_3$ & 3.0 \\
        \hline
    \end{tabular}
    \caption{The weights of the items.}
    \end{subfigure}
    ~
    \par\bigskip
    \begin{subfigure}[t]{\textwidth}
    \centering
    \begin{tabular}{|c|c|c|c|c|}
        \hline
              & $\mathit{TruePositives}$ & $\mathit{FalsePositives}$ & $\mathit{FalseNegatives}$ & $\mathit{TrueNegatives}$  \\
        \hline
        $i_1$ & $\{i_1\}$ &  $\{i_3\}$ &  $\{i_2\}$ &  $\emptyset$ \\
        $i_2$ &  $\{i_2\}$ & $\emptyset$ &  $\{i_1\}$ &  $\{i_3\}$ \\
        $i_3$ & $\{i_3\}$ & $\{i_1\}$ & $\emptyset$ & $\{i_2\}$ \\
        \hline
    \end{tabular}
    \caption{The classification of all items from the perspective of each item.}
    \end{subfigure}
    ~
    \par\bigskip
    \begin{subfigure}[t]{\textwidth}
    \centering
    \begin{tabular}{|c|c|c|c|c|}
        \hline
              & $\mathit{TP}$ & $\mathit{FP}$ & $\mathit{FN}$ & $\mathit{TN}$ \\
        \hline
        $i_1$ & 1 & 3 & 2 & 0 \\
        $i_2$ & 2 & 0 & 1 & 3 \\
        $i_3$ & 3 & 1 & 0 & 2 \\
        \hline
    \end{tabular}
    \caption{The entries of the 2x2 confusion matrix of each item.}
    \end{subfigure}
    ~
    \par\bigskip
    \begin{subfigure}[t]{\textwidth}
    \centering
    \begin{tabular}{|c|c|c|c|}
        \hline
              & $\mathit{Precision}$ & $\mathit{Recall}$ & $\mathit{JaccardDistance}$ \\
        \hline
        $i_1$ & 1/4 (25\%)  & 1/3 (33.33\%)   & 5/6 (83.33\%) \\
        $i_2$ & 1 (100\%) & 2/3 (66.67\%)   & 1/3 (33.33\%) \\
        $i_3$ & 3/4 (75\%)  & 1 (100\%)  & 1/4 (25\%) \\
        \hline
    \end{tabular}
    \caption{Selected clustering metrics from the perspective of each item.}
    \end{subfigure}
    ~
    \par\bigskip
    \begin{subfigure}[t]{\textwidth}
    \centering
    \begin{tabular}{|ll|c|c|c|}
        \hline
             & & $\mathit{Precision}$ & $\mathit{Recall}$ & $\mathit{JaccardDistance}$ \\
        \hline
        Overall & $T = \{i_1, i_2, i_3\}$ & 3/4 (75\%) & 7/9 (77.78\%) & 3/8 (37.5\%) \\
        Ideal cluster & $\{i_1, i_2\}$ & 3/4 (75\%) & 5/9 (55.56\%) & 1/2 (50\%) \\
        Actual cluster & $\{i_1, i_3\}$ & 5/8 (62.5\%) & 5/6 (83.33\%) & 19/48 (39.58\%) \\
        Item slice & $\{i_2, i_3\}$   & 17/20 (85\%) & 13/15 (86.67\%) & 17/60 (28.33\%) \\
        Ideal cluster & $\{i_3\}$       & 3/4 (75\%) & 1 (100\%) & 1/4 (25\%) \\
        \hline
    \end{tabular}
    \caption{Selected clustering metrics for aggregates of items.}
    \end{subfigure}
    \caption{Pointwise clustering metrics in action.}\label{small_clustering_example}
\end{figure*}

% Weighted averages script (python3):

% import fractions

% def WeightedAverage(weights_and_values):
%   numerator = 0
%   denominator = 0
%   for weight, value in weights_and_values:
%     numerator += fractions.Fraction(weight) * fractions.Fraction(value)
%     denominator += fractions.Fraction(weight)
%   return numerator / denominator

% i_1 = {'weight': 1, 'Precision': fractions.Fraction(1,4), 'Recall': fractions.Fraction(1,3), 'JaccardDistance': fractions.Fraction(5,6)}
% i_2 = {'weight': 2, 'Precision': fractions.Fraction(1), 'Recall': fractions.Fraction(2,3), 'JaccardDistance': fractions.Fraction(1,3)}
% i_3 = {'weight': 3, 'Precision': fractions.Fraction(3,4), 'Recall': fractions.Fraction(1), 'JaccardDistance': fractions.Fraction(1,4)}

% def PrintMetrics(item_set):
%   line = ''
%   for metric_name in ['Precision', 'Recall', 'JaccardDistance']:
%     weights_and_values = [(item['weight'], item[metric_name]) for item in item_set]
%     wa = WeightedAverage(weights_and_values)
%     print(f'{metric_name}: {wa}\t({(wa.numerator / wa.denominator):.2%})')
%     line += f' & {wa} ({(wa.numerator / wa.denominator * 100):.2f}\\%)'
%   line += ' \\\\'
%   print(line)

% PrintMetrics([i_1, i_2, i_3])

\subsection{Mathematical properties}

The simplicity of the pointwise definitions make the clustering metrics easy to interpret. They also make the metrics mathematically well-behaved, as is evident from this selection of properties:

\begin{figure*}[t!]
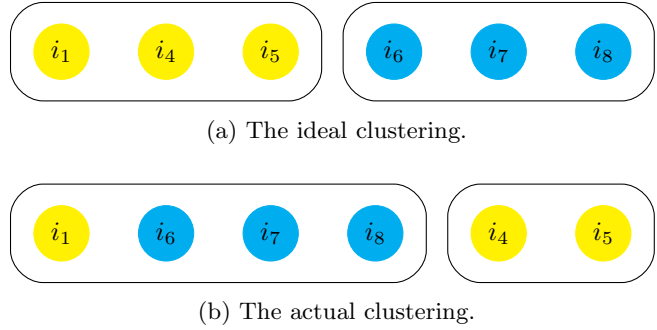

    \centering
    \begin{subfigure}[t]{\textwidth}
        \centering
        \begin{dot2tex}[dot,mathmode]
            graph {
        		subgraph cluster_1 {
        		    style=rounded
        			node [color=yellow, style=filled]
        			i5 [label=i_5]
        			i4 [label=i_4]
        			i1 [label=i_1]
        		}
        		subgraph cluster_2 {
        		    style=rounded
        			node [color=cyan, style=filled]
        			i8 [label=i_8]
        			i7 [label=i_7]
        			i6 [label=i_6]
        		}
        	}
        \end{dot2tex}
        \caption{The ideal clustering.}
    \end{subfigure}
    ~
    \par\bigskip
    \begin{subfigure}[t]{\textwidth}
        \centering
        \begin{dot2tex}[dot,mathmode]
            graph {
        		subgraph cluster_1 {
        		    style=rounded
        			node [color=cyan, style=filled]
        			i8 [label=i_8]
        			i7 [label=i_7]
        			i6 [label=i_6]
        			node [color=yellow, style=filled]
        			i1 [label=i_1]
        		}
        		subgraph cluster_2 {
        		    style=rounded
        			node [color=yellow, style=filled]
        			i5 [label=i_5]
        			i4 [label=i_4]
        		}
        	}
        \end{dot2tex}
        \caption{The actual clustering.}
    \end{subfigure}
    \caption{Clusterings like those of Figure~\ref{small_clustering_example} but where the items have been subdivided: $i_2$ was divided into $i_4$ and $i_5$, while $i_3$ was divided into $i_6$, $i_7$ and $i_8$. If the new items all have weight 1, then overall $\mathit{Precision} = 3/4$, overall $\mathit{Recall} = 7/9$ and overall $\mathit{JaccardDistance} = 3/8$, as was the case in Figure~\ref{small_clustering_example}.}\label{subdivided_items}
\end{figure*}

\begin{enumerate}
\item
Only the relative magnitudes of the weights matter. For example, the weights of the items in Figure~\ref{small_clustering_example} can be multiplied by any positive constant factor, and its overall $\mathit{Precision}$ and $\mathit{Recall}$ and $\mathit{JaccardDistance}$ metrics will remain exactly the same.

\item
Expected values compose nicely. Some instances of this:
\begin{itemize}
    \item The pointwise metric of a singleton set is equal to the same pointwise metric of the item it contains. \\
    Example application in Figure~\ref{small_clustering_example}: The $\mathit{Precision}$ of the actual cluster $\{i_2\}$ is equal to $\mathit{Precision}(i_2) = 1$.
    \item For a set of items $I$, any pointwise metric, such as $\mathit{Precision}(I)$, is equal to the expected value of the same pointwise metric of a partitioning of $I$. \\ For instance, the overall $\mathit{Precision}$ is equal to the weighted average $\mathit{Precision}$ of the clusters in a given clustering (remember that the weight of a cluster is equal to the sum of the weights of its members). \\
    Example application in Figure~\ref{small_clustering_example}: The overall $\mathit{Precision}$ is 3/4. The clusters of the actual clustering are $\{i_1, i_3\}$ and $\{i_2\}$, which have weights 4 and 2 respectively, and $\mathit{Precision}$ metrics of 5/8 and 1 respectively. The weighted average of the clusters' $\mathit{Precision}$ is $\frac{4 \cdot \frac{5}{8} + 2 \cdot 1}{4 + 2} = 3/4$, which is equal to the overall $\mathit{Precision}$.
\end{itemize}

\item
For all $i$, all items in $\mathit{TruePositives}(i)$ will have exactly the same 2x2 confusion matrix and hence the same pointwise metrics.

\item
Subdividing or fusing items in the same $\mathit{TruePositives}$ set won't affect the overall metrics as long as the weight of the set stays constant. \\
Example: The clusterings in Figure~\ref{subdivided_items} are basically those of Figure~\ref{small_clustering_example} where each item has been divided into a number of items equal to its weight. If all the items in Figure~\ref{subdivided_items} have weight 1, then its overall pointwise metrics will be exactly the same as those of Figure~\ref{small_clustering_example}.

% Another property:
% An ideal cluster's pointwise recall metric, and its items' pointwise recall metrics, remain the same if all the items of the ideal cluster have their weight scaled by the same positive factor.

\item
The definitions of $\mathit{Precision}$ and $\mathit{Recall}$ are perfectly symmetric. Given two clusterings $C_1$ and $C_2$, if we consider $C_1$ as the ideal clustering and $C_2$ as the actual clustering and compute the $\mathit{Precision}$ of some set of items $I$, then that will be equal to the $\mathit{Recall}$ of $I$ when the roles of $C_1$ and $C_2$ are swapped, i.e. when $C_2$ is the ideal clustering and $C_1$ is the actual clustering. Formally, we can express this as:
$$
\mathit{Precision}(C_1, C_2) = \mathit{Recall}(C_2, C_1)
$$
where $\mathit{Precision}(C_1, C_2)$ is the function that maps $I$ to its $\mathit{Precision}$ when $C_1$ is treated as the ideal clustering and $C_2$ is treated as the actual clustering.

\item
$\mathit{JaccardDistance}$ is a true distance metric on the set of all clusterings of a fixed set of weighted items. \\
Stated otherwise: Let $C_1$, $C_2$, $C_3$ denote clusterings of the same weighted items. Let $\mathit{JaccardDistance}(C_1, C_2)$ denote the overall $\mathit{JaccardDistance}$ between clustering $C_1$ and clustering $C_2$. We have:
\begin{itemize}
    \item Identity of indiscernibles: \\ $\mathit{JaccardDistance}(C_1, C_2) = 0 \Longleftrightarrow C_1 = C_2$
    \item Symmetry: \\ $\mathit{JaccardDistance}(C_1, C_2) = \mathit{JaccardDistance}(C_2, C_1)$
    \item Triangle inequality: \\ $\mathit{JaccardDistance}(C_1, C_2) + \mathit{JaccardDistance}(C_2, C_3) \geq \mathit{JaccardDistance}(C_1, C_3)$
\end{itemize}
Proof outline: $\mathit{SymmetricDifference}(I_1, I_2) = \mathit{weight}((I_1 \setminus I_2) \cup (I_2 \setminus I_1))$ is a true distance metric on $P(T)$, i.e. sets of items, and by the Steinhaus Transform $\mathit{JaccardDistance}(I_1, I_2)$ is consequently a true distance metric on $P(T)$. 
Let $C_1$ and $C_2$ be two clusterings of $T$, and let $C_1(i)$ denote the cluster of $C_1$ that contains the item $i$. It is quite easy to prove that the lifted $\mathit{JaccardDistance}(C_1, C_2)$, defined as the weighted average of $\mathit{JaccardDistance}(i) = \mathit{JaccardDistance}(C_1(i), C_2(i))$ over all $i \in T$, is a true distance metric on clusterings. See Appendix~\ref{appendix_jaccard_distance_true_metric} for details.

The result is general: \textit{every} true distance metric on sets of weighted items can be lifted to clusterings, where the lifted metric is equal to the weighted average (i.e. expected value) of the original metric over all the items.

\end{enumerate}

\subsection{Practical considerations}

In practice, a ground truth clustering of a set of weighted items can be constructed once, stored, and used many times to evaluate actual clusterings. The big benefit of such a setup is that the human judgements are effectively reused, so their cost is amortized, and they facilitate quick feedback during experimentation (no new judgements are needed to get metrics for a new actual clustering).

In many practical applications, the population of items that must be clustered changes over time, and hence the items it has in common with the ground truth clustering can also change over time. For such applications it makes sense to report auxiliary metrics, such as the number and weight of the items that are common, and/or the number and weight of the items in the ground truth clustering that are not present in the actual clustering, in order to communicate the degree to which the ground truth clustering is still applicable.

A ground truth clustering effectively acts as a specification, and the clustering metrics characterize the degree to which the actual clustering meets the specification. In practice the specification can be very partial and omits the vast majority of the actual clustering's items. Some ground truth clusterings explicitly target a particular slice of items. It is important to understand that the metrics don't necessarily extrapolate well to items that are not included in the specification.

The ground truth clustering provides a point of reference for evaluation purposes. Suppose we want to evaluate an A/B test between a Baseline clustering algorithm and an Experiment clustering algorithm. To do that, we can run both on the same population of items, and obtain a Baseline actual clustering and an Experiment actual clustering. Next, we can evaluate each of these actual clusterings against the same ground truth clustering. It acts as a point of reference, so now we can investigate and compare how the two algorithms processed the items of a given ground truth cluster. We will typically also want to know which algorithm is better. For that, we can look at the delta metrics, for example $\Delta\mathit{Precision} = \mathit{Precision}_\mathrm{Experiment} - \mathit{Precision}_\mathrm{Baseline}$. If there are two more variants of the clustering algorithm we would like to evaluate, then we can perform two more evaluations with the ground truth clustering. Because the ground truth clustering acts as a point of reference, we can easily obtain delta metrics between all 6 pairs of algorithms with a grand total of only 4 evaluations with the ground truth clustering.

Storing the pointwise metrics for the individual items can facilitate easy aggregation after an evaluation run. That makes it easy to drill into the metrics, for example to understand for which items the mistakes happened, and to report metrics for particular slices of items.

\section{Related work}

There are many definitions of clustering quality and similarity metrics in the academic literature. There are also nice overview articles, for example~\cite{entity_resolution_evaluation_measures,stackexchange_measures_to_compare_clustering_partitions}. Here, we briefly discuss the relationship with two large groups of metrics outlined in~\cite{stackexchange_measures_to_compare_clustering_partitions}: metrics based on a co-membership confusion matrix, and metrics based on a frequency cross-tabulation of items. These groups include many popular metrics, such as Adjusted Rand (ARAND) and the V-measure.

\subsection{Metrics based on a co-membership confusion matrix}

The co-membership confusion matrix considers how pairs of items are co-occur in clusters in the ideal and actual clusterings. It constructs a single 2x2 confusion matrix on the basis of that information. In particular, it first classifies pairs of distinct items as follows:
\begin{itemize}
    \item $\mathrm{TruePositives}$: the pairs of items that share an ideal cluster and are also in the same actual cluster.
    \item $\mathrm{FalseNegatives}$: the pairs of items that share an ideal cluster but are in separate actual clusters.
    \item $\mathrm{FalsePositives}$: the pairs of items that are in separate ideal clusters, but in the same cluster in the actual clustering.
    \item $\mathrm{TrueNegatives}$: the pairs of items that are in separate ideal clusters and also in separate actual clusters.
\end{itemize}
It then creates a 2x2 confusion matrix by using $|\mathrm{TruePositives}|$, etc. for the entries, and then it can compute a host of metrics from the confusion matrix.

The confusion matrix only mentions pairs of distinct items. Hence, for example, an item that is a in a singleton cluster in both clusterings (i.e. an item that is in a cluster on its own in the ideal clustering and in the actual clustering) will participate only in $\mathrm{TrueNegatives}$ and never in $\mathrm{TruePositives}$. So it cannot influence $\mathit{Precision}$ at all. A big cluster of size $n$ in the ideal clustering that is preserved in the actual clustering will contribute $O(n^2)$ pairs of items to $\mathrm{TruePositives}$, so it can have an amplified influence on $\mathit{Precision}$.

The pointwise metrics of this paper can also be expressed in terms of pairs (not only pairs of distinct items -- self-pairs also play a role), but the set-theoretical presentation above is far simpler and easier to understand. It assigns metrics directly to individual items, which makes it easy to report metrics for aggregates such as slices or clusters. It deals well with singleton clusters, and the influence of a large cluster on the overall metrics is regulated by its weight, which is the sum of the weights of its items, as one would expect.

We note that any desirable metric definition of the co-membership confusion matrix approach can be used directly in the pointwise approach: every metric that is defined on a 2x2 confusion matrix, such as Adjusted Rand (ARAND), can be computed for each item, and lifted to aggregates of items with expected values.

\subsection{Metrics based on a frequency cross-tabulation}

The second family of metrics is based on a cross-tabulation of items. The standard formulation uses frequencies, but we can easily generalize it to use weights. In the following matrix, the rows represent ideal clusters and the columns represent actual clusters, and the entry $w_{ij}$ denotes the weight of the intersection of ideal cluster $i$ and actual cluster $j$:
$$
\begin{bmatrix}
w_{11} & w_{12} & \cdots & w_{1j} & \cdots & w_{1M} \\
w_{21} & w_{22} & \cdots & w_{2j} & \cdots & w_{2M} \\
\vdots & \vdots & \ddots & \vdots &        & \vdots \\
w_{i1} & w_{i2} & \cdots & w_{ij} & \cdots & w_{iM} \\
\vdots & \vdots &        & \vdots & \ddots & \vdots \\
w_{N1} & w_{N2} & \cdots & w_{Nj} & \cdots & w_{NM} 
\end{bmatrix}
$$

The pointwise metrics of this paper can be formulated in terms of this matrix: every item in the intersection of ideal cluster $i$ and actual cluster $j$ has a 2x2 confusion matrix, which summarizes the clustering situation from its perspective, specified by:
\begin{itemize}
    \item $\mathit{TP} = w_{ij}$
    \item $\mathit{FP} = \left(\sum_{h = 1}^N w_{hj}\right) - w_{ij}$
    \item $\mathit{FN} = \left(\sum_{k = 1}^M w_{ik}\right) - w_{ij}$
    \item $\mathit{TN} = \left(\sum_{h = 1}^N\sum_{k = 1}^M w_{hk}\right) - \mathit{TP} - \mathit{FP} - \mathit{FN}$
\end{itemize}
The per-item metrics are then defined on the basis of its 2x2 confusion matrix, and lifted to sets of items using expected values (weighted averages).

Many metrics in the literature are based on a cross-tabulation of items. Among them, the F-measure~\cite{fung2003hierarchical} and the B-CUBED~\cite{bagga1998algorithms} family are the most closely related to the pointwise metrics of this paper. We discuss each of them in turn, and also briefly touch upon the V-measure~\cite{rosenberg-hirschberg-2007-v}.

The F-measure~\cite{fung2003hierarchical}, also known as `F Clustering Accuracy', is spot on with its definitions of $\mathit{Precision}_{ij}$ and $\mathit{Recall}_{ij}$. It computes $F_{ij}$ as the harmonic mean of $\mathit{Precision}_{ij}$ and $\mathit{Recall}_{ij}$. The $F_{ij}$ values are then aggregated: the ideal cluster $h$ is associated with a score, namely $\max_{j \in 1..M}\{F_{hj}\}$, and the F-measure of the overall clustering is defined as the weighted average of the scores of the ideal clusters. So, in contrast to the pointwise metrics, it does not define a 2x2 confusion matrix for $ij$, it doesn't compute per-item metrics, and its aggregation is different.

B-CUBED~\cite{bagga1998algorithms}, also written $B^3$, is another related approach based on a cross-tabulation of items. It assigns each item a $\mathit{Precision}$ and $\mathit{Recall}$ value, where the definitions of $\mathit{Precision}$ and $\mathit{Recall}$ are the same as the pointwise ones of this paper when items are unweighted (i.e. all items have an equal weight). It then computes overall $\mathit{Precision}$ and $\mathit{Recall}$ metrics as the weighted averages of the per-item metrics. So it uses weighted averages in the aggregation, but weights don't play a role in the per-item metrics. It does not make the full 2x2 confusion matrix explicit for each item, and other natural metrics for 2x2 confusion matrices, such as $\mathit{Accuracy}$ and $\mathit{JaccardDistance}$, are not mentioned.

The V-measure~\cite{rosenberg-hirschberg-2007-v} is a popular metric for measuring the similarity of two clusterings. It is also based on a cross-tabulation of items. It is defined as the (weighted) harmonic mean of $\mathit{Homogeneity}$ and $\mathit{Completeness}$, whose definitions are based on information theory. The conceptual counterparts in the pointwise metrics are $\mathit{Precision}$ and $\mathit{Recall}$. The pointwise formulation makes it straightforward to support also other metrics of 2x2 confusion matrices, such as $\mathit{Informedness}$ and $\mathit{Markedness}$~\cite{powers2011evaluation} and $\mathit{Accuracy}$, and the pointwise $\mathit{JaccardDistance}$ is a true distance metric for clusterings. It is not immediately clear whether the V-measure's information-theoretic framework can support such notions, and whether it can supply per-item metrics that compose nicely.

\section{Conclusion}

There are many definitions of clustering quality and similarity in the academic literature. The definitions we discussed in this paper allow us to:
\begin{itemize}
    \item Associate weights, or importances, to items, and to accommodate that in the metrics.
    \item Associate metrics with individual items, clusters, and arbitrary slices, which is useful for debugging and drilling deeper. The metrics are simple to understand and intuitive.
    \item Use the core theory as a foundation of clustering similarity and quality. The metrics are mathematically well-behaved, which makes it easy to build more things on top of them, for example techniques to obtain statistical estimates of the quality of huge clusterings. This paper discussed in depth how the pointwise metrics can be used to evaluate an actual clustering with respect to a ground truth clustering, which is typically small but nonetheless a useful indicator of quality.
\end{itemize}

\bibliographystyle{plain}
\bibliography{main}

\appendix

\section{Proof that the $\mathit{JaccardDistance}$ is a true distance metric for clusterings}\label{appendix_jaccard_distance_true_metric}

This appendix contains a proof that the $\mathit{JaccardDistance}$ is a true distance metric for clusterings. 

We have a non-empty set of items $T$, where each item $i \in T$ is associated with a positive real $\mathit{weight}(i)$. For a set of items $I \subseteq T$, let $\mathit{weight}(I) = \sum_{i \in I} \mathit{weight}(i)$.

To get things started, we define:

$$\mathit{SymmetricDifference}(I_1, I_2) = \mathit{weight}((I_1 \setminus I_2) \cup (I_2 \setminus I_1))$$

\begin{lemma}
$(P(T), \mathit{SymmetricDifference})$ is a metric space, i.e. it satisfies the following properties:
\begin{enumerate}
    \item (Zero distance to self) $\mathit{SymmetricDifference}(I, I) = 0$.
    \item (Positivity) $I_1 \neq I_2$ implies $\mathit{SymmetricDifference}(I_1, I_2) > 0$.
    \item (Symmetry) $\mathit{SymmetricDifference}(I_1, I_2) = \mathit{SymmetricDifference}(I_2, I_1)$.
    \item (Triangle inequality) $\mathit{SymmetricDifference}(I_1, I_2) + \mathit{SymmetricDifference}(I_2, I_3) \geq \mathit{SymmetricDifference}(I_1, I_3)$.
\end{enumerate}
\end{lemma}

\begin{proof}
We prove each property in turn:
\begin{enumerate}
    \item $\mathit{SymmetricDifference}(I, I) = \mathit{weight}(\emptyset) = 0$.
    \item $I_1 \neq I_2$ implies $(I_1 \setminus I_2) \cup (I_2 \setminus I_1)$ is non-empty, and since all weights are positive, we have $\mathit{SymmetricDifference}(I_1, I_2) > 0$.
    \item $\mathit{SymmetricDifference}$ is symmetric because \\ $(I_1 \setminus I_2) \cup (I_2 \setminus I_1) = (I_2 \setminus I_1) \cup (I_1 \setminus I_2)$.
    \item $\mathit{SymmetricDifference}(I_1, I_2) + \mathit{SymmetricDifference}(I_2, I_3)$ \\
    $=$ \\
    $\mathit{weight}((I_1 \setminus I_2) \cup (I_2 \setminus I_1)) + \mathit{weight}((I_2 \setminus I_3) \cup (I_3 \setminus I_2))$ \\
    $\geq$ \\
    $\mathit{weight}((I_1 \setminus I_2) \cup (I_2 \setminus I_1) \cup (I_2 \setminus I_3) \cup (I_3 \setminus I_2))$ \\
    $=$ [Best seen by looking at the Venn diagram of $I_1$, $I_2$ and $I_3$] \\
    $\mathit{weight}((I_1 \cup I_2 \cup I_3) \setminus (I_1 \cap I_2 \cap I_3))$ \\
    $\geq$ [Again best seen by looking at the Venn diagram of $I_1$, $I_2$ and $I_3$] \\
    $\mathit{weight}((I_1 \setminus I_3) \cup (I_3 \setminus I_1))$ \\
    $=$ \\
    $\mathit{SymmetricDifference}(I_1, I_3)$
\end{enumerate}
\end{proof}

Next, we define the $\mathit{JaccardDistance}$ between two sets of items as follows:

$$\mathit{JaccardDistance}(I_1, I_2) = 1 - \frac{\mathit{weight}(I_1 \cap I_2)}{\mathit{weight}(I_1 \cup I_2)}$$

\begin{theorem}\label{jaccard_distance_on_sets_is_a_metric}
$(P(T), \mathit{JaccardDistance})$ is a metric space.
\end{theorem}

\begin{proof}
It follows from the Steinhaus Transform, which can be stated as:

A metric space $(X, d)$ and a constant point $p \in X$ induces a new metric $d^\prime$ defined as
$$d^\prime(x, y) = \frac{2d(x,y)}{d(x,p) + d(y,p) + d(x,y)}$$

For the proof at hand, we can use $d = \mathit{SymmetricDifference}$ and $p = \emptyset$ and expand the definition of $d^\prime$ to get:
\begin{align*}
 & \hphantom{=l} \frac{2 \cdot \mathit{SymmetricDifference}(x, y)}{\mathit{weight}(x) + \mathit{weight}(y) + \mathit{SymmetricDifference}(x, y)} \\
 & = \frac{2 \cdot \mathit{weight}((x \setminus y) \cup (y \setminus x))}{\mathit{weight}(x) + \mathit{weight}(y) + \mathit{weight}((x \setminus y) \cup (y \setminus x))} \\
 & = \frac{2 \cdot \mathit{weight}((x \setminus y) \cup (y \setminus x))}{\mathit{weight}(x) + \mathit{weight}(y) + \mathit{weight}(x \setminus y) + \mathit{weight}(y \setminus x)} \\
 & = \text{[Best seen by looking at the Venn diagram of $x$ and $y$]} \\
 & \hphantom{=l} \frac{2 \cdot \mathit{weight}((x \setminus y) \cup (y \setminus x))}{2 \cdot \mathit{weight}(x \cup y)} \\
 & =  \frac{\mathit{weight}(x \cup y) - \mathit{weight}(x \cap y)}{\mathit{weight}(x \cup y)} \\
 & =  1 - \frac{\mathit{weight}(x \cap y)}{\mathit{weight}(x \cup y)} \\
 & =  \mathit{JaccardDistance}(x, y)
\end{align*}
\end{proof}

Let $\mathrm{Clusterings}(T)$ denote the set of all clusterings of items in $T$, and let $C$, possibly sub-scripted, denote a clustering. Let $C(i) \subseteq T$ denote the cluster in $C$ that contains item $i$.

Every true distance metric on sets of weighted items can be lifted to clusterings, where the lifted metric is equal to the weighted average (i.e. the expected value), of the original metric over all items:

\begin{theorem}\label{lifting_of_true_metrics}
If $(P(T), d)$ is a metric space, then so is $(\mathrm{Clusterings}(T), d^\prime)$, where
$$d^\prime(C_1, C_2) = \frac{\sum_{i \in T} \mathit{weight}(i) \cdot d(C_1(i), C_2(i))}{\mathit{weight}(T)}$$
\end{theorem}

\begin{proof}
Assume $(P(T), d)$ is a metric space, i.e.
\begin{enumerate}
    \item (Zero distance to self) $d(I, I) = 0$.
    \item (Positivity) $I_1 \neq I_2$ implies $d(I_1, I_2) > 0$.
    \item (Symmetry) $d(I_1, I_2) = d(I_2, I_1)$.
    \item (Triangle inequality) $d(I_1, I_2) + d(I_2, I_3) \geq d(I_1, I_3)$.
\end{enumerate}
For $d^\prime$ we have to show:
\begin{enumerate}
    \item (Zero distance to self) $d^\prime(C, C) = 0$.
    \item (Positivity) $C_1 \neq C_2$ implies $d^\prime(C_1, C_2) > 0$.
    \item (Symmetry) $d^\prime(C_1, C_2) = d^\prime(C_2, C_1)$.
    \item (Triangle inequality) $d^\prime(C_1, C_2) + d^\prime(C_2, C_3) \geq d^\prime(C_1, C_3)$.
\end{enumerate}
We prove each property in turn:
\begin{enumerate}
    \item $$d^\prime(C, C) = \frac{\sum_{i \in T} \mathit{weight}(i) \cdot d(C(i), C(i))}{\mathit{weight}(T)} = \frac{\sum_{i \in T} \mathit{weight}(i) \cdot 0}{\mathit{weight}(T)} = 0$$
    \item Assume $C_1 \neq C_2$. Then for at least one $i \in T$ it must be the case that $C_1(i) \neq C_2(i)$, and from the Positivity of $d$ we know $d(C_1(i), C_2(i)) > 0$.  So the numerator of $d^\prime(C_1, C_2)$ will be positive and hence $d^\prime(C_1, C_2) > 0$.
    \item Symmetry holds trivially because $d$ is symmetric.
    \item \begin{align*}
              & \hphantom{=l} d^\prime(C_1, C_2) + d^\prime(C_2, C_3)  \\
              & = \frac{\sum_{i \in T} \mathit{weight}(i) \cdot d(C_1(i), C_2(i))}{\mathit{weight}(T)} + \frac{\sum_{i \in T} \mathit{weight}(i) \cdot d(C_2(i), C_3(i))}{\mathit{weight}(T)} \\
              & = \frac{\sum_{i \in T} \mathit{weight}(i) \cdot [d(C_1(i), C_2(i)) + d(C_2(i), C_3(i))]}{\mathit{weight}(T)} \\
              & \geq \frac{\sum_{i \in T} \mathit{weight}(i) \cdot d(C_1(i), C_3(i))}{\mathit{weight}(T)} \\
              & = d^\prime(C_1, C_3)
          \end{align*}
\end{enumerate}
\end{proof}

Hence we have the main result:

\begin{theorem}
$(\mathrm{Clusterings}(T), \mathit{JaccardDistance})$ is a metric space, where
$$\mathit{JaccardDistance}(C_1, C_2) = \frac{\sum_{i \in T} \mathit{weight}(i) \cdot \mathit{JaccardDistance}(C_1(i), C_2(i))}{\mathit{weight}(T)}$$
\end{theorem}

\begin{proof}
It is a direct consequence of Theorems~\ref{jaccard_distance_on_sets_is_a_metric} and \ref{lifting_of_true_metrics}.
\end{proof}

\end{document}